\documentclass[12pt,reqno]{amsart}

\newcommand\version{August 19, 2018}

\usepackage{amsmath,amsfonts,amsthm,amssymb,amsxtra}
\usepackage{bbm} % to get \1

\newtheorem{theorem}{Theorem}%[section]
\newtheorem{proposition}[theorem]{Proposition}

\newtheorem{corollary}[theorem]{Corollary}

\theoremstyle{definition}

\theoremstyle{remark}

\newtheorem{remark}[theorem]{Remark}

\newcommand{\1}{\mathbbm{1}}

\newcommand{\const}{\mathrm{const}\ }

\renewcommand{\epsilon}{\varepsilon}

\newcommand{\loc}{{\rm loc}}

\newcommand{\N}{\mathbb{N}}

\renewcommand{\phi}{\varphi}
\newcommand{\R}{\mathbb{R}}

\newcommand{\Z}{\mathbb{Z}}

\DeclareMathOperator{\per}{Per}

\begin{document}

\title[A one-dimensional liquid drop model --- \version]{Periodic energy minimizers\\ for a one-dimensional liquid drop model}

\author{Rupert L. Frank}
\address[Rupert L. Frank]{Mathematisches Institut, Ludwig-Maximilans Universit\"at M\"unchen, Theresienstr. 39, 80333 M\"unchen, Germany, and Mathematics 253-37, Caltech, Pasa\-de\-na, CA 91125, USA}
\email{rlfrank@caltech.edu}

\author{Elliott H. Lieb}
\address[Elliott H. Lieb]{Departments of Mathematics and Physics, Princeton University, Washington Road, Princeton, NJ 08544, USA}
\email{lieb@princeton.edu}

\thanks{\copyright\, 2018 by the authors. This paper may be
reproduced, in its entirety, for non-commercial purposes.\\
U.S.~National Science Foundation grants DMS-1363432 (R.L.F.) and PHY-1265118 (E.H.L.) are acknowledged.}

\begin{abstract}
We reprove a result by Ren and Wei concerning the periodicity of minimizers of a one-dimensional liquid drop model in the neutral case. Our proof works for general boundary conditions and also in the non-neutral case.
\end{abstract}

\maketitle

\section{Introduction and main result}

In this paper we consider the energy functional
\begin{equation}
\label{eq:energy}
\mathcal I_\rho^{(L)}[E] :=
\per E - \frac{\gamma}{2} \int_{-L/2}^{L/2} \int_{-L/2}^{L/2} (\1_E(x)-\rho)|x-y|(\1_E(y)-\rho)\,dx\,dy
\end{equation}
defined on sets $E\subset[-L/2,L/2]$ and involving a parameter $\rho\in (0,1)$, as well as the corresponding ground state energy
\begin{equation}
\label{eq:gse}
e^{(L)}_\rho(\ell) := \inf\left\{ \mathcal I_\rho^{(L)}[E] :\ E\subset [-L/2,L/2],\, |E|=\ell\right\}.
\end{equation}
The constant $\gamma>0$ is fixed throughout this paper and will not be reflected in the notation. (In fact, by rescaling $E$ and $L$ we could set $\gamma=1$.) By $\per E$ we denote the perimeter of the set $E$ in the sense of geometric measure theory which, however, is elementary in this one-dimensional context. Namely, a bounded set $E\subset\R$ is of finite perimeter if and only if, up to sets of measure zero, there is an $N\in\N$ such that $E$ is the union of $N$ intervals whose closures are disjoint, and in this case $\per E = 2N$.

The minimization problem \eqref{eq:gse} arises in nuclear physics. As suggested originally in \cite{HSY,RPW} nuclear matter at extremely high densities, as for instance, in the crust of neutron stars, exhibits exotic phases, sometimes called `nuclear pasta phases'. The relevant parameter $\rho\in (0,1)$ describes the ratio between the charge density of a uniform background of electrons and that of the nuclei. For values of $\rho$ around $1/2$ it is believed that nuclear matter arranges itself in a slab-like structure which is periodic with respect to one direction. Within Gamow's liquid drop model \cite{G} this slab-like regime is described by the energy functional \eqref{eq:energy}.

The model \eqref{eq:energy}, however, is of interest also beyond this concrete physical problem. It is variant of a one-dimensional Coulomb problem. These are introduced as toy models which mimick some of the properties of the (much harder) three-dimensional Coulomb problem and have been studied, for instance, in \cite{L,B,K,BL,JJ}. One phenomenon which is of particular interest is the emergence of periodic structures. While a proof of this property still eludes us in the three-dimensional context, it has been shown to occur in several one-dimensional models; see, for instance, \cite{M,RW,RW2,CO,GLL} and references therein.

Remarkably, the minimization problem defining $e^{(L)}_\rho(\ell)$ can be solved explicitly. In the `neutral' case $\ell=\rho L$ this was shown in a different, but essentially equivalent formulation in the work \cite{RW} by Ren and Wei. We give a more quantitative and, we think, simpler proof of their solution. Moreover, we present several extensions which, we believe, are new. One of these concerns the study of the non-neutral case $\ell=\rho L+Q$ with an excess charge $Q\neq 0$. We show that this excess charge goes to the boundary and \emph{lowers} the energy per length (in the thermodynamic limit $L\to\infty$) by an amount of $\gamma Q^2/4$. This is in contrast to the three-dimensional case, where the excess charge \emph{raises} the energy per volume by an amount proportional to $Q^2$ \cite{LL}.

Another generalization concerns the Coulomb kernel $-\frac12|x-y|$ in \eqref{eq:energy}. This function coincides, up to an irrelevant additive constant, with the Neumann Green's function on the interval $(-L/2,L/2)$. (Because of this fact our result in the neutral case is equivalent to the Ren--Wei result.) In other occurrences of the above model, and also as a technical tool in certain proofs, it is natural to consider Green's functions on $(-L/2,L/2)$ with different boundary conditions, namely either periodic or Dirichlet boundary conditions. We show that, remarkably, the ground state energies for these various choices all coincide on any given interval. Moreover, the optimizing sets coincide up to translations.

We now proceed to a precise statement of our main results. We begin with the `neutral' case $\ell=\rho L$ considered previously in \cite{RW}. We consider the set
$$
E_{\rho,N,L} = \bigcup_{n=1}^N \left[ \frac{(2n-N-1-\rho)L}{2N},\frac{(2n-N-1+\rho)L}{2N}\right] \,.
$$
This is the union of $N$ intervals of length $L/N$ centered at the point $(2n-N-1)L/(2N)$, $n=1,\ldots,N$.

\begin{theorem}\label{minimization}
Let $\rho\in(0,1)$ and $L>0$. Then
\begin{align*}
e^{(L)}_\rho(\rho L) = L \min_{N \in \N} \left( 2(N/L) + \frac{\gamma}{12} \frac{\rho^2(1-\rho)^2}{(N/L)^2} \right).
\end{align*}
The minimum on the right side is attained by at least one and at most two $N\in\N$. Minimizing sets are exactly those of the form $E_{\rho,N,L}$ with a minimizing $N$. In particular, minimizing sets are periodic with minimal period $L/N$.
\end{theorem}

Strictly speaking, this is not exactly the result from \cite{RW}. They consider the energy functional \eqref{eq:energy} with $\per E$ replaced by the relative perimeter $\per (E,(-L/2,L/2))$, where boundaries of $E$ coinciding with one of the points $\pm L/2$ are not counted. This has the effect that their functional has twice as many minimizers.

Having Theorem \ref{minimization} it is easy to compute the thermodynamic limit.

\begin{corollary}\label{tdlim}
Let $\rho\in(0,1)$. Then
\begin{align*}
& \lim_{L\to\infty} \frac{e^{(L)}_\rho(\rho L)}{L} = \left( \frac{3}{2} \right)^{2/3} \gamma^{1/3} \rho^{2/3}(1-\rho)^{2/3} \,.
\end{align*}
Moreover, the set of limit points in $L^1_\loc(\R)$ of minimizers consists of the two sets
$$
\sum_{n\in\Z} \left[\beta \left(n -\frac\rho2\right),\beta\left(n+\frac\rho2\right) \right]
\qquad\text{and}\qquad
\sum_{n\in\Z} \left[\beta \left(n+\frac{1-\rho}2\right),\beta\left(n+\frac{1+\rho}2\right)\right]
$$
with $\beta = 2^{2/3} 3^{1/3} \gamma^{-1/3} (\rho(1-\rho))^{-2/3}$.
\end{corollary}

In fact, we prove that Theorem \ref{minimization} implies the uniform bound
\begin{equation}
\label{eq:tduniform}
\frac{e^{(L)}_\rho(\rho L)}{L} \geq \left( \frac{3}{2} \right)^{2/3} \gamma^{1/3} \rho^{2/3}(1-\rho)^{2/3}
\qquad\text{for all}\ L>0 \,,
\end{equation}
as well as the remainder bound
\begin{equation}
\label{eq:tdrem}
\frac{e^{(L)}_\rho(\rho L)}{L} = \left( \frac{3}{2} \right)^{2/3} \gamma^{1/3}
 \rho^{2/3}(1-\rho)^{2/3} + O(L^{-2})
 \qquad\text{as}\ L\to\infty \,.
\end{equation}
It is remarkable that the remainder here is $O(L^{-2})$ and not $O(L^{-1})$. We also show that the error bound $O(L^{-2})$ cannot be improved.

It is also remarkable that the energy in the thermodynamic limit does not behave linearly as $\rho\to 0$ or $\rho\to 1$. This reflects the fact that the minimization problem $\per E - (\gamma/2) \iint_{E\times E}|x-y|\,dx\,dy$ over sets $E\subset\R$ with fixed $|E|$ yields $-\infty$. In contrast, in the three-dimensional case, where the corresponding whole space problem does have a minimizer, the analogous energy in the thermodynamic limit can be shown to behave linearly as $\rho\to 0$ with a coefficient depending on the whole space problem \cite{EFK}.

Next, we comment on the non-neutral case. Since the explicit solution is somewhat complicated to state, we content ourselves with the statement in the thermodynamic limit.

\begin{corollary}\label{excess}
Let $\rho\in(0,1)$, $L>0$ and $Q\in\R$. Then
$$
\lim_{L\to\infty} L^{-1} e^{(L)}_\rho(\rho L + Q) = \left( \frac{3}{2} \right)^{2/3} \gamma^{1/3}
 \rho^{2/3}(1-\rho)^{2/3} - \frac{1}{4}\gamma Q^2 \,.
$$
\end{corollary}

Thus, non-neutrality lowers the energy per length. We refer to the proof for a description of minimizing sets.

So far, we have considered the problem where the sets interact through the whole space Green's function $-|x-y|/2$. As a final topic, we consider various choices of Green's functions corresponding to different boundary conditions, namely,
\begin{align}
\label{eq:kernels}
& -\frac12 |x-y| + \const & \text{Neumann case} \notag \\
& -\frac12|x-y| - \frac1L xy + \frac14 L & \text{Dirichlet case}\\
& -\frac12|x-y| - \frac1L xy + \const & \text{periodic case} \,. \notag
\end{align}
The constants in the Neumann and in the periodic case are chosen such that the integral of the kernel vanishes (with respect to $x$ for every $y$). Since we apply this kernel only to functions with integral zero, the value of these constants is irrelevant for us.

We denote by $k$ any one of these three kernels and consider the energy functional
$$
\tilde{\mathcal I}_\rho^{(L)}[E] :=
\per E - \gamma \int_{-L/2}^{L/2} \int_{-L/2}^{L/2} (\1_E(x)-\rho)k(x,y)(\1_E(y)-\rho)\,dx\,dy
$$
and the minimization problem
$$
\tilde e^{(L)}_\rho(\ell) := \inf\left\{ \tilde{\mathcal I}_\rho^{(L)}[E] :\ 
E\subset[-L/2,L/2],\ |E|=\ell \right\}.
$$
In the periodic case we agree to interpret $\per E$ as the perimeter of $E$ considered as a subset of $\R/L\Z$ and drop the constraint $E\subset[-L/2,L/2]$, interpreting the double integral as an integral over $(\R/L\Z)\times(\R/L\Z)$.

\begin{theorem}\label{minimizationgen}
Let $k$ be one of the kernels in \eqref{eq:kernels}. Then for any $\rho\in(0,1)$ and $L>0$,
\begin{align}
\label{eq:minimizationgen}
\tilde e^{(L)}_\rho(\rho L) = L \min_{N \in \N} \left( 2(N/L) + \frac{\gamma}{12} \frac{\rho^2(1-\rho)^2}{(N/L)^2} \right) \,.
\end{align}
Moreover, equality holds if and only if
\begin{enumerate}
\item[(1)] in the Neumann case, $E=E_{\rho,N,L}$,
\item[(2)] in the Dirichlet case, $E=E_{\rho,N,L}+a$ for $a\in [-(1+\rho)L/(2N),(1+\rho)L/(2N)]$,
\item[(3)] in the periodic case, $E=E_{\rho,N,L}+a$ for $a\in\R$,
\end{enumerate}
where, in all cases, $N$ is optimal for the minimum on the right side in \eqref{eq:minimizationgen}.
\end{theorem}

The results in the Dirichlet and in the periodic case seem to be new. Non-sharp bounds in the periodic case have been obtained in \cite{O}.

The structure of this paper is as follows. Section \ref{sec:inequality} contains the main inequality on which our argument hinges and we use it to derive Theorem \ref{minimization}. In Section \ref{sec:boundaryconditions} we discuss different boundary conditions and prove Theorem \ref{minimizationgen}. Finally, in Section \ref{sec:tdlimit} we discuss the thermodynamic limit proving Corollary \ref{tdlim}, the bounds stated thereafter and Corollary \ref{excess}.

%%%%%%%%%%%%%%%%%%%%%%%%%

\section{The main inequality}\label{sec:inequality}

The key ingredient in the proof of Theorem \ref{minimization} is the following lower bound.

\begin{proposition}\label{bound}
\begin{enumerate}
\item
Let $\rho\in(0,1)$ and $N\in\N$. For any set $E\subset\R$ which is the union of at most $N$ intervals one has
$$
-\frac12 \iint_{E\times E} |x-y|\,dx\,dy + \rho \int_E x^2\,dx \geq - \frac{1}{12\rho} |E|^3 \left( 1- \frac{(1-\rho)^2}{N^2} \right) \,.
$$
Equality holds if and only if $E$ is the union of exactly $N$ intervals, centered at the points $\frac{(2n-N-1)|E|}{2\rho N}$, $n=1,\ldots,N$, and all of equal length.
\item
Let $\rho\geq 1$. For any set $E\subset\R$ one has
$$
-\frac12 \iint_{E\times E} |x-y|\,dx\,dy + \rho \int_E x^2\,dx \geq \frac{\rho-2}{12} |E|^3 \,.
$$
Equality holds if and only if $E$ is an interval centered at the origin.
\end{enumerate}
\end{proposition}

\begin{proof}
We will prove the assertion of part (1), but with the case $\rho=1$ included. Before doing so, let us observe that this will also imply the statement for $\rho>1$. Indeed, once the $\rho=1$ statement is proved, we know that
$$
-\frac12 \iint_{E\times E} |x-y|\,dx\,dy + \int_E x^2\,dx \geq - \frac{1}{12}|E|^3
$$
with equality if and only if $E$ is an interval centered at the origin. On the other hand, by a simple rearrangement inequality we know that for $\rho>1$
$$
(\rho - 1) \int_E x^2 \,dx \geq (\rho - 1) \int_{-|E|/2}^{|E|/2} x^2 \,dx = \frac{\rho - 1}{12} |E|^3
$$
with equality if and only if $E$ is an interval centered at the origin. This implies the claimed statement for $\rho>1$.

Thus, in the following we will assume that $\rho\in(0,1]$. We denote by $x_1<\ldots<x_N$ the centers of the intervals and by $q_1,\ldots,q_N$ their length. (If there are less than $N$ intervals, we set some of the $q_n$'s equal to zero.) We will show that
\begin{align}\label{eq:formula}
-\frac12 \iint_{E\times E} |x-y|\,dx\,dy + \rho \int_E x^2\,dx
& = \rho \sum_n q_n \left(x_n - \frac{1}{2\rho} \left( \sum_{m<n} q_m - \sum_{m>n} q_m \right) \right)^2 \notag \\
& \qquad + \frac{(1-\rho)^2}{12\rho} \sum_n q_n^3 
- \frac{1}{12\rho}  \left( \sum_n q_n \right)^3 \,.
\end{align}
Dropping the first term on the right side, which is non-negative, and bounding using H\"older's inequality
$$
\sum_n q_n^3 \geq N^{-2} \left( \sum_n q_n \right)^3 \,,
$$
we obtain from \eqref{eq:formula} the lower bound in the proposition. Moreover, the non-negative term that we dropped vanishes if and only if
$$
x_n = \frac{1}{2\rho} \left( \sum_{m<n} q_m - \sum_{m>n} q_m \right)
\qquad\text{for all}\ n=1,\ldots,N \,.
$$
Note that this minimizing configuration is consistent with coming from centers of intervals, since (recalling that $\rho\leq 1$)
$$
x_{n+1} - x_n = \frac{1}{2\rho} (q_n + q_{n+1}) \geq \frac{1}{2} (q_n+q_{n+1}) \,,
$$
so $x_n + q_n/2 \leq x_{n+1}- q_{n+1}/2$. Finally, in H\"older's inequality equality holds if and only if $q_n = N^{-1} \sum_m q_m$ for all $n$. From this we deduce the conditions for equality in the proposition.

It remains to prove identity \eqref{eq:formula}. By a straightforward computation of integrals we find
$$
-\frac12 \iint_{E\times E} |x-y|\,dx\,dy + \rho \int_E x^2\,dx
= - \sum_{n<m} q_n q_m |x_n-x_m| - \frac{2-\rho}{12} \sum_n q_n^3 + \rho \sum_n q_n x_n^2 \,.
$$
Recalling that the $x_n$ are ordered, we can complete the square and obtain
\begin{align*}
- \sum_{n<m} q_n q_m |x_n-x_m| + \rho \sum_n q_n x_n^2
& = \rho \sum_n q_n \left(x_n - \frac{1}{2\rho} \left( \sum_{m<n} q_m - \sum_{m>n} q_m \right) \right)^2 \\
& \qquad - \frac{1}{4\rho} \sum_n q_n \left( \sum_{m<n} q_m - \sum_{m>n} q_m \right)^2 \,.
\end{align*}
We now observe that
$$
\sum_n q_n \left( \sum_{m<n} q_m - \sum_{m>n} q_m \right)^2 + \frac13 \sum_n q_n^3 = \frac13 \left( \sum_n q_n \right)^3 \,.
$$
This can be proved by induction, for instance. Combining the last two identities we obtain \eqref{eq:formula}.
\end{proof}

\begin{corollary}\label{background}
Let $\rho\in(0,1)$, $L>0$ and $N\in\N$. For any set $E\subset [-L/2,L/2]$ which is the union of at most $N$ intervals one has
\begin{align*}
& -\frac12 \int_{-L/2}^{L/2} \int_{-L/2}^{L/2} (\1_E(x)-\rho)|x-y|(\1_E(y)-\rho)\,dx\,dy \\
& \qquad \geq -\frac{1}{12\rho} |E|^3 \left( 1- \frac{(1-\rho)^2}{N^2} \right) + \frac14 \rho |E| L^2- \frac{1}{6} \rho^2 L^3 \,.
\end{align*}
Moreover, if $(N-1+\rho)|E| \leq \rho N L$, then equality holds if and only if $E$ is the union of exactly $N$ intervals, centered at the points $\frac{(2n-N-1)|E|}{2\rho N}$, $n=1,\ldots,N$, and all of equal length.
\end{corollary}

\begin{proof}
Since $E\subset[-L/2,L/2]$ we have
\begin{align*}
& -\frac12 \int_{-L/2}^{L/2} \int_{-L/2}^{L/2} (\1_E(x)-\rho)|x-y|(\1_E(y)-\rho)\,dx\,dy\\
& \qquad = -\frac12 \iint_{E\times E} |x-y|\,dx\,dy + \rho \int_E x^2\,dx + \frac{1}{4} \rho |E| L^2 - \frac1{6} \rho^2 L^3 \,.
\end{align*}
The claimed inequality now follows from the proposition. Moreover, the equality conditions in the proposition are consistent with the constraint $E\subset[-L/2,L/2]$ if and only if $(N-1)|E|/(2\rho N) + |E|/(2N) \leq L/2$.
\end{proof}

Now we are in position to prove our main result.

\begin{proof}[Proof of Theorem \ref{minimization}]
Sets of finite perimeter in $\R$ are finite unions of intervals. Therefore we can compute the infimum over all sets $E$ of finite perimeter with $|E|=\rho L$ by first minimizing over all set $E$ with $|E|=\rho L$ which are the union of exactly $N$ intervals and then taking the infimum over $N$. If we insert $|E|=\rho L$ into the bound in Corollary \ref{background}, we obtain for any set $E\subset[-L/2,L/2]$ with $|E|=\rho L$ which is the union of $N$ intervals,
$$
-\frac12 \int_{-L/2}^{L/2} \int_{-L/2}^{L/2} (\1_E(x)-\rho)|x-y|(\1_E(y)-\rho)\,dx\,dy \geq \frac{1}{12} L^3 \frac{\rho^2(1-\rho)^2}{N^2} \,.
$$
Moreover for such $E$, $\per E=2N$. This yields the claimed lower bound. This lower bound is, in fact, optimal since in the case $|E|=\rho L$ the condition in Corollary \ref{background} is satisfied and therefore the bound is attained by the set described in the corollary.
\end{proof}

%%%%%%%%%%%%%%%%%%%%

%%%%%%%%%%%%%%%%%%%%

\section{Different boundary conditions}\label{sec:boundaryconditions}

Our goal in this section is to prove Theorem \ref{minimizationgen}. The main ingredient in the proof is the following analogue of Proposition \ref{bound} where translation invariance is restored.

\begin{proposition}\label{boundgen}
Let $\rho\in(0,1]$ and $N\in\N$. For any set $E\subset\R$ which is the union of at most $N$ intervals one has
$$
-\frac12 \iint_{E\times E} |x-y|\,dx\,dy + \rho \int_E x^2\,dx - \frac{\rho}{|E|} \left( \int_E x\,dx \right)^2
\geq - \frac{1}{12\rho} |E|^3 \left( 1- \frac{(1-\rho)^2}{N^2} \right) \,.
$$
Equality holds if and only if $E$ is the union of exactly $N$ intervals, centered at the points $\frac{(2n-N-1)|E|}{2\rho N}+X$, $n=1,\ldots,N$, for some $X\in\R$ and all of equal length.
\end{proposition}

\begin{proof}
Let $X:= |E|^{-1} \int_E x\,dx$ and $E'=E-X$. Then
\begin{align*}
& -\frac12 \iint_{E\times E} |x-y|\,dx\,dy + \rho \int_E x^2\,dx - \frac{\rho}{|E|} \left( \int_E x\,dx \right)^2 \\
& \qquad = -\frac12 \iint_{E\times E} |x-y|\,dx\,dy + \rho \int_E (x-X)^2\,dx \\
& \qquad =-\frac12 \iint_{E'\times E'} |x-y|\,dx\,dy + \rho \int_{E'} x^2\,dx \,.
\end{align*}
Since $|E'|=|E|$ and since $E'$ is also the union of at most $N$ intervals, the claimed lower bound now follows immediately from Proposition \ref{bound}.

Moreover, also by that proposition, equality holds if and only if $E'$ is the union of exactly $N$ intervals, centered at the points $\frac{(2n-N-1)|E|}{2\rho N}$, $n=1,\ldots,N$, and all of equal length. Clearly this is equivalent to the statement in the proposition.
\end{proof}

\begin{proof}[Proof of Theorem \ref{minimizationgen}]
In the Neumann case we have $\tilde{\mathcal I}^{(L)}_\rho[E] = \mathcal I^{(L)}_\rho[E]$ provided $|E|=\rho L$, so the assertion follows immediately from Theorem \ref{minimization}.

In the Dirichlet and the periodic case we have for $|E|=\rho L$ that
$$
\tilde{\mathcal I}^{(L)}_\rho[E] = \mathcal I^{(L)}_\rho[E] - \frac1L \left( \int_{-L/2}^{L/2} x(\1_E(x)-\rho)\,dx \right)^2 = \mathcal I^{(L)}_\rho[E] - \frac{\rho}{|E|} \left( \int_E x\,dx \right)^2 \,.
$$
Arguing as in the proof of Corollary \ref{background}, with Proposition \ref{boundgen} instead of Proposition \ref{bound}, we obtain the assertion.
\end{proof}

%%%%%%%%%%%%%%%%%%%%

\section{The thermodynamic limit}\label{sec:tdlimit}

With the exact formula from Theorem \ref{minimization} at hand it is easy to compute the thermodynamic limit with optimal remainder estimates.

\begin{proof}[Proof of Corollary \ref{tdlim}]
We use the explicit expression for the infimum from Theorem \ref{minimization} and write
$$
2(N/L) + \frac{\gamma}{12} \frac{\rho^2(1-\rho)^2}{(N/L)^2} = \gamma^{1/3} \rho^{2/3} (1-\rho)^{2/3} \ f\left( \frac{2N}{L\gamma^{1/3} \rho^{2/3} (1-\rho)^{2/3}}\right)
$$
with
$$
f(x) = x + \frac{1}{3x^2} \,.
$$
The function $f$ has a unique minimum at $x=(2/3)^{1/3}$ with $f((2/3)^{1/3}) = (3/2)^{2/3}$.
\end{proof}

\begin{remark}
Since $f(x)\geq (3/2)^{2/3}$ for all $x$, the preceding proof yields the uniform bound \eqref{eq:tduniform}. Moreover, along the sequence $(L_N)_{N\in\N}$ defined by $2N/(L_N\gamma^{1/3} \rho^{2/3}(1-\rho)^{2/3}) = (2/3)^{1/3}$ we have
$$
e_\rho^{(L)}(\rho L_N) - L_N \left( \frac{3}{2} \right)^{2/3} \gamma^{1/3}
 \rho^{2/3}(1-\rho)^{2/3} = 0 \,.
$$
We now prove the remainder bound \eqref{eq:tdrem} and show its optimality. Given $L>0$ choose $N$ such that
$$
\frac{2N}{L\gamma^{1/3} \rho^{2/3} (1-\rho)^{2/3}} \leq (2/3)^{1/3} < \frac{2(N+1)}{L\gamma^{1/3} \rho^{2/3} (1-\rho)^{2/3}}
$$
and define
$$
\delta_- = (2/3)^{1/3} - \frac{2N}{L\gamma^{1/3} \rho^{2/3} (1-\rho)^{2/3}} \,,
\qquad
\delta_+ = \frac{2(N+1)}{L\gamma^{1/3} \rho^{2/3} (1-\rho)^{2/3}} -  (2/3)^{1/3} \,.
$$
Then
$$
0\leq \delta_\pm \leq \frac{2}{L\gamma^{1/3} \rho^{2/3} (1-\rho)^{2/3}} \,.
$$
Since  the function $f$ introduced in the previous proof has a unique local minimum,
\begin{align*}
L^{-1} e_\rho^{(L)}(\rho L) & = \gamma^{1/3} \rho^{2/3} (1-\rho)^{2/3} \min\left\{ f\left( \frac{2N}{L\gamma^{1/3} \rho^{2/3} (1-\rho)^{2/3}} \right), f\left( \frac{2(N+1)}{L\gamma^{1/3} \rho^{2/3} (1-\rho)^{2/3}} \right) \right\} \\
& = \gamma^{1/3} \rho^{2/3} (1-\rho)^{2/3} \min\left\{ f\left( (2/3)^{1/3}-\delta_- \right), f\left( (2/3)^{1/3} + \delta_+ \right) \right\}.
\end{align*}
Since
$$
f(x) = (3/2)^{2/3} + c (x- (2/3)^{1/3})^2 + o((x-(2/3)^{1/3})^2)
\qquad\text{as}\ x\to (2/3)^{1/3}
$$
with $c=(3/2)^{4/3}$, we conclude that
$$
L^{-1} e_\rho^{(L)}(\rho L) = \gamma^{1/3} \rho^{2/3} (1-\rho)^{2/3} \left( (3/2)^{2/3} + c \min\{ \delta_+^2 +o(\delta_+^2), \delta_-^2 + o(\delta_-^2)\} \right).
$$
Clearly,
$$
\limsup_{L\to\infty} L \min\{\delta_+,\delta_-\} = \frac{1}{\gamma^{1/3} \rho^{2/3}(1-\rho)^{2/3}}
$$
and therefore
$$
\limsup_{L\to\infty} L^2 \left( L^{-1} e_\rho^{(L)}(\rho L) - (3/2)^{2/3}\gamma^{1/3} \rho^{2/3} (1-\rho)^{2/3} \right) = \frac{c}{\gamma^{1/3} \rho^{2/3}(1-\rho)^{2/3}} \,.
$$
This proves the claimed optimal error bound.
\end{remark}

Finally, we discuss the problem with an excess charge.

\begin{proof}[Proof of Corollary \ref{excess}]
We infer from Corollary \ref{background} that for any set $E\subset[-L/2,L/2]$ with $|E|=\rho L +Q$ which consists of at most $N$ intervals we have the lower bound
\begin{align*}
\mathcal I^{(L)}_\rho[E] & \geq 2N -\frac{\gamma}{12\rho} (\rho L +Q)^3 \left( 1- \frac{(1-\rho)^2}{N^2} \right) + \frac\gamma 4 \rho (\rho L+Q) L^2- \frac{\gamma}{6} \rho^2 L^3 \\
& = 2N + \frac{\gamma}{12\rho} (\rho L +Q)^3 \frac{(1-\rho)^2}{N^2} -\frac{\gamma}{12\rho} (3 \rho L Q^2 + Q^3) \,.
\end{align*}
Therefore,
\begin{align}\label{eq:excesslower}
e^{(L)}_\rho(\rho L+Q) \geq \min_{N\in\N} \left( 2N + \frac{\gamma}{12\rho} (\rho L +Q)^3 \frac{(1-\rho)^2}{N^2}\right) -\frac{\gamma}{12\rho} (3 \rho L Q^2 + Q^3).
\end{align}
Clearly,
$$
\lim_{L\to\infty} L^{-1} \frac{\gamma}{12\rho} (3 \rho L Q^2 + Q^3) = \frac{\gamma}{4} Q^2 \,,
$$
which gives the claimed contribution to the energy due to the excess charge. Moreover, elementary analysis shows that
\begin{equation}
\label{eq:elem}
\lim_{L\to\infty} L^{-1} \min_{N\in\N} \left( 2N + \frac{\gamma}{12\rho} (\rho L +Q)^3 \frac{(1-\rho)^2}{N^2}\right) = \left( \frac{3}{2} \right)^{2/3} \gamma^{1/3} \rho^{2/3}(1-\rho)^{2/3} \,.
\end{equation}
This yields the claimed asymptotic lower bound
$$
\liminf_{L\to\infty} L^{-1} e^{(L)}_\rho(\rho L+Q) \geq \left( \frac{3}{2} \right)^{2/3} \gamma^{1/3} \rho^{2/3}(1-\rho)^{2/3} - \frac{\gamma}{4} Q^2 \,.
$$

In order to prove an asymptotic upper bound we first assume
\begin{equation}
\label{eq:assq}
Q<2^{2/3} 3^{1/3}\gamma^{-1/3} \rho^{1/3} (1-\rho)^{1/3} \,.
\end{equation}
In fact, under this assumption we will be able to solve the $e^{(L)}_\rho(\rho L + Q)$ problem explicitly for $L$ large enough. To do so, we note that the elementary analysis leading to \eqref{eq:elem} shows also that the minimum on the right side is attained by some $N$ satisfying
$$
N = 2^{-2/3} 3^{-1/3} \gamma^{1/3} \rho^{2/3} (1-\rho)^{2/3} L +o(L) \,. 
$$
Therefore, for $L$ large enough we can restrict the minimum in \eqref{eq:excesslower} to such $N$, and then assumption \eqref{eq:assq} implies that the inequality
$$
(N-1+\rho) Q \leq \rho(1-\rho)L
$$
holds for all considered $N$. Using the latter inequality, we infer from the second part of Corollary \ref{background} that the above lower bound on $\mathcal I^{(L)}_\rho[E]$ can be saturated and therefore we infer that equality holds in \eqref{eq:excesslower} for all sufficiently large $L$. This proves the claimed asymptotic upper bound under the assumption \eqref{eq:assq}.

It remains to deal with $Q$ for which \eqref{eq:assq} does not hold. In fact, we give a proof that works for all $Q>0$ by reducing it to the case $Q<0$ (and $\rho$ to $1-\rho$). This proof, however, does not yield the optimal set. We start by observing
$$
\mathcal I^{(L)}_\rho[E] = \mathcal I^{(L)}_{1-\rho}[(-L/2,L/2)\setminus E] + \left( \per E - \per \left( (-L/2,L/2)\setminus E \right) \right).
$$
Since
$$
\left| \per E - \per \left( (-L/2,L/2)\setminus E \right) \right| \leq 2 \,,
$$
we conclude that for all $\ell>0$,
$$
\left| e^{(L)}_\rho(\ell) - e^{(L)}_{1-\rho}(L-\ell) \right| \leq 2 \,.
$$
In particular, because of what we have already shown in the first part of the proof (noting that this formula is invariant under changing $\rho$ to $-\rho$ and $Q$ to $-Q$),
\begin{align*}
L^{-1} e^{(L)}_\rho(\rho L+Q) & = L^{-1} e^{(L)}_{1-\rho}((1-\rho)L-Q) + O(L^{-1}) \\
& = \left( \frac{3}{2} \right)^{2/3} \gamma^{1/3} \rho^{2/3}(1-\rho)^{2/3} - \frac{\gamma}{4} Q^2 + o(1) \,.
\end{align*}
This proves the claimed asymptotic upper bound.
\end{proof}

\bibliographystyle{amsalpha}

\end{document}